\documentclass[journal]{IEEEtran}

\usepackage{cite}
\usepackage{amsmath,amssymb,amsfonts,amsthm}
\usepackage[ruled,lined,linesnumbered]{algorithm2e}
\usepackage{graphicx,subfig}
\usepackage{enumitem}
\usepackage{tabularx,multirow,booktabs,makecell}

\def\BibTeX{{\rm B\kern-.05em{\sc i\kern-.025em b}\kern-.08em
    T\kern-.1667em\lower.7ex\hbox{E}\kern-.125emX}}
\SetCommentSty{textnormal}
\SetArgSty{textnormal}

\newcommand{\equref}[1]{(\ref{#1})}

\newcommand{\figref}[1]{Fig. \ref{#1}}
\newcommand{\algref}[1]{Algorithm \ref{#1}}

\newtheorem{proposition}{Proposition}

\begin{document}

\bstctlcite{BSTcontrol}
\title{Successive Cancellation Decoding For General Monotone Chain Polar Codes}
\author{Zichang~Ren, Chunhang~Zheng, Dou~Li, Yuping~Zhao}
\maketitle

\begin{abstract}
    Monotone chain polar codes generalize classical polar codes to multivariate settings, offering a flexible approach for achieving the entire admissible rate region in the distributed lossless coding problem. However, this flexibility also introduces significant challenges for existing successive cancellation (SC) based decoding schemes. Motivated by the need for a general SC decoding solution, we present a comprehensive decoding strategy for monotone chain polar codes that can handle arbitrary numbers of terminals, non-binary alphabets, and decoding along arbitrary monotone chains. Specifically, we formulate the SC decoding task as a series of inference subtasks over the polar transform and propose a computational graph framework based on probability propagation principles. This approach highlights the impact of variable switching during decoding and shows that time complexity varies between $O(N\log{N})$ and $O(N^2)$, depending on the specific chain structure. Moreover, we demonstrate that the widely used $O(N)$ space optimization is not universally applicable to monotone chain polar codes, which prompts us to introduce a constant-time decoder forking strategy based on the proposed logical computation graphs. This strategy enables time-efficient list decoding without relying on $O(N)$-space techniques. Numerical results verify the superior performance of the proposed scheme compared with the classical lazy-copy scheme.
\end{abstract}


\section{Introduction}

Polar codes, introduced by Arıkan \cite{5075875}, achieve the capacity of binary-input discrete memoryless symmetric channels and have since inspired extensive research on polarization theory. In various classical scenarios, such as lossless \cite{5513567} and lossy \cite{5437372} source coding, as well as multiterminal settings \cite{6874846}, polar codes have been proven theoretically optimal.

Among these developments, distributed lossless source coding, also known as the Slepian-Wolf (SW) problem \cite{1055037}, has attracted significant attention due to its broad practical relevance. Prior studies \cite{5454148,5503184,6208869} explored the use of polar codes in multiple access channels, the channel-coding dual of the SW problem. Arıkan initially showed that corner points of the SW region can be achieved using polar codes with side information \cite{5513567}, and the full dominant face can be covered via time-sharing. In a subsequent work \cite{6284254}, he proposed a more insightful approach by generalizing polarization to multivariate chain rules of joint entropy. By exploring a class of monotone chains, he demonstrated that polar codes can be constructed to reach any point in the SW rate region by following different monotone chains. We refer to such constructions as \emph{monotone chain polar codes}.

These codes extend classical polar codes to multivariate sources. Many of their properties including construction, encoding, and polarization rate naturally follow from the classical case. However, decoding introduces new technical challenges. A practical decoding algorithm was proposed in \cite{6620401}, where the authors introduced four recursive calculation rules for the so-called coordinate channels, enabling low-complexity decoding based on conditional probability distributions.

Existing decoding strategies, however, leave several key issues unaddressed. Related works \cite{6620401,7282710} focus on a narrow class of chains similar to source-splitting schemes \cite{613189}, which represent only a small subset of all monotone chains. As a result, their conclusions lack generality. It can be problematic in certain cases. For instance, we find that not all monotone chain polar codes support the well-known $O(N)$ space optimization \cite{7055304}. For these chains, the decoder forking operation based on standard lazy-copy strategy requires $O(N)$ time, leading to an overall complexity of $O(N^2)$ instead of $O(N\log{N})$. This makes list decoding costly for moderate or large block lengths. We also find that, even for successive cancellation (SC) decoding, the widely known $O(N\log{N})$ time complexity is not always attainable for certain chains. Moreover, existing algorithms are limited to binary, two-terminal scenarios. Motivated by applications involving multiple terminals and non-binary alphabets \cite{quantum,10578043}, there is a clear need for a general and efficient decoding framework.

In this paper, we present the complexity analysis and implementation for SC decoding of general monotone chain polar codes. We formulate the SC decoding task as a sequence of inference tasks over polar transforms, and propose a computational graph framework for theoretical analysis. This framework enables the design of time-efficient SC and list decoding algorithms for arbitrary numbers of correlated discrete memoryless sources with arbitrary alphabets along arbitrary monotone chains. We show that the decoding time complexity ranges from $O(N\log{N})$ to $O(N^2)$, and both bounds are tight. We further demonstrate that the classical $O(N)$ space optimizations are neither universally applicable nor fundamentally necessary. Leveraging the proposed structure, we develop a constant-time decoder forking strategy that enables efficient list decoding without relying on $O(N)$ space techniques. Such refinements are essential for general monotone chains where conventional optimizations fail to apply.

\section{Preliminaries}

\subsection{Notation Conventions}

Throughout this paper, we adopt several notational conventions to enhance clarity. Random variables are denoted by uppercase letters, such as $X$, while their realizations in a particular experiment are represented by lowercase letters, such as $x$. Since we study multiple correlated sources and their repeated independent uses, superscripts are used to distinguish different sources, and subscripts indicate different temporal instances of the same source. For example, $X^1_1, X^1_2$ and $X^2_1, X^2_2$ denote two independent copies of the sources $X^1$ and $X^2$, respectively. Here, $X^1_1,X^2_1$ represent two different but correlated sources, and the same for $X^1_2,X^2_2$.

When a colon appears in a superscript or subscript, it denotes a sequence specified by its endpoints. For instance, we define $X_{1:N} \triangleq X_1,\ldots,X_N$ and $Y^{1:M} \triangleq Y^1,\ldots,Y^M$. To avoid ambiguity, we do not use colon notation simultaneously in both subscripts and superscripts. In particular, if a subscript or superscript is a set, it denotes the subsequence formed by the elements of that set in ascending order. For example, we define $x_{\mathcal{A}} \triangleq x_{a_1},\ldots,x_{a_{\lvert\mathcal{A}\rvert}}$, where $a_i \in \mathcal{A}$ are sorted in increasing order.

\subsection{Polarization and Polar Coding}

In our work, we focus solely on the source formulation of the SW problem. Accordingly, we begin with a brief review of source polarization and source polar coding introduced by \cite{5513567}.

Let $X_{1:N}$ consist of $N = 2^n$ independent copies of a binary memoryless source $X$, where $n \geq 1$. These variables undergo a linear transform:
\begin{equation}
    U_{1:N} = X_{1:N} \mathbf{G}_N,
\end{equation}
where $\mathbf{G}_N = \mathbf{G}_2^{\otimes n}$ is the $n$-th Kronecker power of the standard transform $\mathbf{G}_2 = \begin{pmatrix}1&0\\1&1\end{pmatrix}$. Given the invertibility of $\mathbf{G}_N$ and the chain rule for entropy expansion, we have:
\begin{equation}
    N\cdot H(X) = H(X_{1:N}) = \sum_{i=1}^{N}H(U_i|U_{1:i-1}).
\end{equation}

The phenomenon of source polarization refers to the fact that, for any $\delta \in (0,1)$, as $N \to \infty$, the proportion of indices with intermediate conditional entropy vanishes:
\begin{equation}
    \frac{\lvert \left\{ i : H(U_i|U_{1:i-1}) \in (\delta, 1-\delta) \right\} \rvert}{N} \to 0,
    \label{equ:classic_polar}
\end{equation}
which implies that the conditional entropies $H(U_i|U_{1:i-1})$ polarize toward $0$ or $1$ as the block length increases.

This polarization property enables compression of the source sequence $x_{1:N}$. One begins by selecting a high-entropy set $\mathcal{F} \subset \{1, \dots, N\}$, also referred to as the frozen set, with rate $R$ (i.e., $\lvert\mathcal{F}\rvert = \lceil NR \rceil$), such that $\forall i\in\mathcal{F}$ and $\forall j\not\in\mathcal{F}$, we have $H(U_i|U_{1:i-1}) \geq H(U_j|U_{1:j-1})$. The encoder then performs the polar transform $u_{1:N} = x_{1:N} \mathbf{G}_N$ and outputs the subsequence $u_{\mathcal{F}}$ as codeword.

The decoder estimates each bit sequentially by computing the conditional probabilities $\Pr(U_i = u \mid U_{1:i-1} = \hat{u}_{1:i-1})$ from $i=1$ to $N$, and applying the rule:
\begin{equation*}
    \hat{u}_i = \begin{cases}
        u_i, & \text{if } i \in \mathcal{F} \\
        0, & \text{if } \dfrac{\Pr(U_i=0 \mid U_{1:i-1} = \hat{u}_{1:i-1})}{\Pr(U_i=1 \mid U_{1:i-1} = \hat{u}_{1:i-1})} > 1 \\
        1, & \text{otherwise}
    \end{cases}
\end{equation*}
to make hard decisions. Once all $\hat{u}_{1:N}$ are obtained, the decoder then applies the inverse transform $\hat{x}_{1:N} = \hat{u}_{1:N} \mathbf{G}_N^{-1}$, which yields the reconstructed source sequence. Although decoding errors may occur (i.e., $x_{1:N} \neq \hat{x}_{1:N}$), it can be shown that if $R > H(X)$, the error probability vanishes as $N \to \infty$. Therefore, the scheme achieves the source entropy.

It is worth noting that the polarization behavior described in \eqref{equ:classic_polar} does not necessarily hold for non-binary sources. Nevertheless, prior work \cite{5351487,6283740} has been done to generalize polarization to non-binary alphabets. As observed in \cite{7456294}, it is always possible to induce polarization over non-binary sources by imposing a quasigroup operation on the input alphabet, while preserving nearly all desirable properties of binary polar codes, including polarization rate and encoding/decoding complexity. In this work, we adopt the random permutation strategy proposed in \cite{5351487} for general alphabets, with further details provided in later sections.

\subsection{SC List Decoding}

The theoretical optimality of polar codes under SC decoding only emerges at extremely large block lengths. In practical settings, their performance is limited due to insufficient polarization. To address this, numerous variants of the SC algorithm have been proposed \cite{7055304,stack,7094848,6560025}, many of which can be viewed as heuristic search strategies over the code tree \cite{6560025}.

Among these methods, SC list decoding \cite{7055304} stands out as the most successful and widely adopted. The key idea is to retain up to $L$ decoding candidates at each non-frozen bit, instead of committing to the most likely one. At the end of decoding, the candidate with the highest likelihood is selected as the final output. Owing to its excellent performance at short to moderate block lengths, SC list decoding has been extensively studied. Notable enhancements include integration with cyclic redundancy check \cite{6355936,6297420,7862172}, development of fast implementations \cite{7339671,8669947}, and theoretical analysis \cite{9770084}. The core idea has also inspired advances in related areas such as belief propagation list decoding \cite{8396299} and list decoding of polarization-adjusted convolutional codes \cite{9174118,9328621}. Without doubt, list decoding plays a central role in the practical adoption and standardization of polar codes.

When list decoding is used, computational efficiency becomes a key concern. A naive implementation, where each new decoder after forking copies all $O(N\log{N})$ internal data, results in an overall decoding complexity of $O(N^2\log{N})$, which is clearly impractical. To overcome this, a standard optimization was proposed in \cite{7055304} as follows:
\begin{itemize}
    \item Store all necessary data in categorized probability and bit arrays of size $O(N)$;
    \item Maintain a set of $O(\log{N})$ pointers to access these arrays;
    \item During each decoder forking operation, copy only the pointers rather than the data.
\end{itemize}

It is clear that the low complexity of standard list decoding heavily relies on the $O(\log{N})$-complexity lazy-copy operation under space optimization. Unfortunately, as we will show in this work, such optimization is not generally applicable for monotone chains and appears to be a fortunate consequence of the classical SC decoding order.

\subsection{Monotone Chain Polar Codes}

One of the theoretical foundations of source polarization is the chain rule of entropy, which can be extended to multivariate random variables. Consider $M$ correlated discrete memoryless sources $X^{1:M}$ with respective alphabet sizes $q_{\gamma}\geq 2,\gamma=1,\ldots,M$. Let $X^1_{1:N},\ldots,X^M_{1:N}$ be $N = 2^n$ independent copies of each source. Applying the polar transform to each source yields:
\begin{equation}
    U^\gamma_{1:N} = X^\gamma_{1:N} \mathbf{G}_N.
\end{equation}

Given the invertibility of $\mathbf{G}_N$ and the chain rule for joint entropy, we have:
\begin{equation}
    N\cdot H(X^{1:M}) = \sum_{t=1}^{MN}H(U^{\gamma_t}_{i_t} \mid U^{\gamma_1}_{i_1},\ldots,U^{\gamma_{t-1}}_{i_{t-1}}),
\end{equation}
where the sequence $U^{\gamma_1}_{i_1},\ldots,U^{\gamma_{MN}}_{i_{MN}}$ can be any permutation of the sequence $U^{1:M}_1,\ldots,U^{1:M}_{N}$.

In particular, \cite{6284254} considers a class of chains with monotonicity, where the natural order within the $N$ copies of each source is preserved. More precisely, for any $t < t'$, if $\gamma_t = \gamma_{t'}$, then it must hold that $i_t < i_{t'}$. This condition ensures that the decoding order is consistent with the natural polarization structure of the sources. It is worth noting that such monotonicity enables any monotone chain to be uniquely identified by its superscript sequence $\gamma_{1:MN}$. We will frequently employ this notation in the following.

For any given monotone chain $\gamma_{1:MN}$, we define its $k$-extension as:
\begin{equation*}
    \underbrace{\gamma_1,\ldots,\gamma_1}_{2^k},\ \ldots,\ \underbrace{\gamma_{MN},\ldots,\gamma_{MN}}_{2^k},
\end{equation*}
and it can be shown that, for any $\delta\in(0,1)$, as $k \to \infty$, the $k$-extension satisfies the following property :
\begin{equation}
    \frac{\left| \left\{ t : H(U^{\gamma_t}_{i_t} \mid U^{\gamma_1}_{i_1},\ldots,U^{\gamma_{t-1}}_{i_{t-1}}) \in (\delta, 1-\delta) \right\} \right|}{2^k MN} \to 0,
\end{equation}
where $t \in \{1, \ldots, 2^k MN\}$. Although this result is originally proved in \cite{6284254} for the simplest case $M = 2$, $q_1 = q_2 = 2$, it can be extended naturally to arbitrary numbers of terminals $M \geq 2$ with arbitrary alphabet sizes $q_{\gamma} \geq 2$ for $\gamma = 1,\ldots,M$ by leveraging the results in \cite{quantum,5351487,6283740}.

This polarization effect along monotone chains enables a distributed lossless source coding scheme analogous to that of \cite{5513567}, by transmitting only high-entropy symbols and discarding the low-entropy ones. Furthermore, different chains correspond to different rate allocations across terminals. It is also shown in \cite{6284254,quantum} that for any target point on the dominant face of the SW rate region, there exists a monotone chain whose induced rate vector approximates it arbitrarily closely. This allows the entire admissible region to be achieved without relying on time-sharing.

\section{Fundamentals and Complexity Analysis of SC Decoding}

The SC decoding problem for classical polar codes is conceptually simple and direct. It involves computing the following conditional probabilities from $t = 1$ to $N$:
\begin{equation*}
    \Pr(U_t | U_{1:t-1} = u_{1:t-1}).
\end{equation*}

As noted in \cite{5075875}, due to the recursive structure of the polar transform, the complexity of computing each probability is no more than $O(N)$. Hence, the total complexity is at most $O(N^2)$. If intermediate results from various decoding steps are shared, the total complexity can be further reduced to the efficient $O(N \log N)$.

\begin{figure}[t!]
    \centerline{\includegraphics[width=0.45\textwidth]{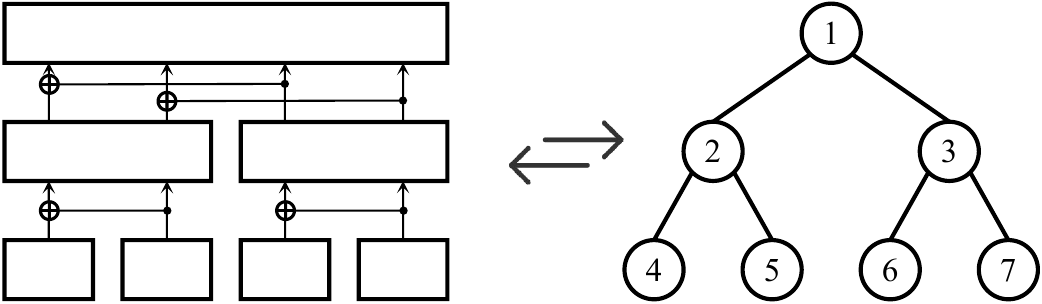}}
    \caption{Binary tree representation of the size-$2^2$ polar transform.}
    \label{fig:decTree}
\end{figure}

In relevant studies, this algorithm is often intuitively explained as a probability propagation process over a binary tree \cite{6852102}. Specifically, the polar transform belongs to the class of so-called butterfly transforms \cite{butterfly}, where a size-$N$ transform can be constructed by connecting two size-$N/2$ transforms with $N/2$ size-2 transforms. The corresponding relationship is shown in \figref{fig:decTree}. For a code of length $N = 2^n$, the decoding tree is a complete binary tree with $n + 1$ layers. The $j$-th layer contains $2^{j-1}$ nodes, each associated with a group of $N / 2^{j-1}$ multivariate random variables. Let the variable group at a parent node be $P^{1:M}_1, \ldots, P^{1:M}_{2l}$, and those at its left and right children be $L^{1:M}_1, \ldots, L^{1:M}_l$ and $R^{1:M}_1, \ldots, R^{1:M}_l$, respectively. Here, $l$ denotes the number of multivariate variables in each child group. Then, for all $i = 1, \ldots, l$ and $\gamma = 1, \ldots, M$, the following relations hold:
\begin{equation}
    P^\gamma_i = L^\gamma_i + R^\gamma_i,\quad P^\gamma_{i+l} = R^\gamma_i,
    \label{equ:polar_relation}
\end{equation}
which can be further refined according to \cite{5351487} to induce polarization in the non-binary case:
\begin{equation}
    P^\gamma_i = L^\gamma_i + R^\gamma_i,\quad P^\gamma_{i+l} = \pi(R^\gamma_i),
\end{equation}
where $\pi(\cdot)$ represents a random mapping over $\mathbb{Z}_{q_i}$. In practice, a fixed permutation may be used. Since this modification does not alter the core algorithmic structure, we focus on \equref{equ:polar_relation} in the following discussions.

In the SC decoding of classical polar codes, we encounter three fundamental subproblems:
\begin{enumerate}
    \item Given the distributions of $P_i$ and $P_{i+l}$, compute the distribution of $L_i$;
    \item Given the distributions of $P_i$ and $P_{i+l}$, along with the known value of $L_i$, compute the distribution of $R_i$;
    \item Given the known values of $L_i$ and $R_i$, determine the values of $P_i$ and $P_{i+l}$.
\end{enumerate}

These subproblems have been thoroughly studied. In the binary case, the first two are solved by:
\begin{equation*}
    \mathcal{L}_{L_i} = \frac{\mathcal{L}_{P_i} \cdot \mathcal{L}_{P_{i+l}} + 1}{\mathcal{L}_{P_i} + \mathcal{L}_{P_{i+l}}},\quad \mathcal{L}_{R_i} = \mathcal{L}_{P_i}^{1 - 2k} \cdot \mathcal{L}_{P_{i+l}},
\end{equation*}
where $\mathcal{L}$ denotes the likelihood ratio, e.g., $\mathcal{L}_{L_i} = \Pr(L_i = 0)/\Pr(L_i = 1)$, and $k$ is the known value of $L_i$. The third subproblem can be solved by direct algebraic computation.

In fact, if viewed from the binary tree representation of polar transforms, the above subproblems are actually special cases of the following generalized subproblems:
\begin{enumerate}
    \item Given the distributions of $P^{1:M}_i$, $P^{1:M}_{i+l}$, and $R^{1:M}_i$, compute the distribution of $L^{1:M}_i$;
    \item Given the distributions of $P^{1:M}_i$, $P^{1:M}_{i+l}$, and $L^{1:M}_i$, compute the distribution of $R^{1:M}_i$;
    \item Given the distributions of $L^{1:M}_i$ and $R^{1:M}_i$, compute the distributions of $P^{1:M}_i$ and $P^{1:M}_{i+l}$.
\end{enumerate}

For instance, considering $R_i$ is unknown in classical subproblem 1, we may assume it follows a uniform distribution. Similarly, when a variable is known deterministically, such as $L_i = k$, its distribution is a delta function: $\Pr(L_i = k) = 1$ and zero elsewhere.

In classical polar codes, there is no need to consider the general subproblems. Due to the so-called "partial order", the uncertainty of $L_i$ is always greater than that of $R_i$, so it is not possible to freeze $R_i$ without freezing $L_i$. However, this is not the case for monotone chain polar codes. In these codes, some components in $L^{1:M}_i$ and $R^{1:M}_i$ may have known values, meaning the simple partial order no longer holds. As a result, we need to consider more general cases.

\subsection{Tensor Representation and Basic Computations}

To express the subsequent calculation more concisely, we represent the joint distribution of a multivariate variable $X^{1:M}$ by an $M$-dimensional tensor $\mathcal{P}^{X^{1:M}}$, with entries:
\begin{equation}
    \mathcal{P}^{X^{1:M}}_{x_{1:M}} \triangleq \Pr(X^{1:M} = x_{1:M}).
\end{equation}

We define three basic tensor operations. The circular convolution is denoted by $\circledast$. If $\mathcal{P}^{X^{1:M}} = \mathcal{P}^{Y^{1:M}} \circledast \mathcal{P}^{Z^{1:M}}$, then:
\begin{equation}
    \mathcal{P}^{X^{1:M}}_{k_{1:M}} = \sum_{i_1=0}^{q_1 - 1} \cdots \sum_{i_M=0}^{q_M - 1} \mathcal{P}^{Y^{1:M}}_{k_{1:M} - i_{1:M}} \cdot \mathcal{P}^{Z^{1:M}}_{i_{1:M}},
    \label{equ:circonv}
\end{equation}
which computes the distribution of the sum of two multivariate random variables. Since modular addition is not self-inverse in general, we define a dual convolution $\circledast'$ such that:
\begin{equation}
    \mathcal{P}^{X^{1:M}}_{k_{1:M}} = \sum_{i_1=0}^{q_1 - 1} \cdots \sum_{i_M=0}^{q_M - 1} \mathcal{P}^{Y^{1:M}}_{k_{1:M} + i_{1:M}} \cdot \mathcal{P}^{Z^{1:M}}_{i_{1:M}},
    \label{equ:circonv_inv}
\end{equation}
to compute the distribution of the difference between multivariate variables. We also define the symbol $\odot$ to denote the normalized elementwise product. If $\mathcal{P}^{X^{1:M}} = \mathcal{P}^{Y^{1:M}} \odot \mathcal{P}^{Z^{1:M}}$, then:
\begin{equation}
    \mathcal{P}^{X^{1:M}}_{k_{1:M}} = \frac{\mathcal{P}^{Y^{1:M}}_{k_{1:M}} \cdot \mathcal{P}^{Z^{1:M}}_{k_{1:M}}}{\sum_{i_1=0}^{q_1 - 1} \cdots \sum_{i_M=0}^{q_M - 1} \mathcal{P}^{Y^{1:M}}_{i_{1:M}} \cdot \mathcal{P}^{Z^{1:M}}_{i_{1:M}}},
    \label{equ:nrmcomb}
\end{equation}
which is used to combine multiple independent observations of the same multivariate variable.

With these notations, the three subproblems can be expressed compactly. For subproblem 1, we have:
\begin{equation}
    \mathcal{P}^{L^{1:M}_i} = \mathcal{P}^{P^{1:M}_i} \circledast \left(\mathcal{P}^{P^{1:M}_{i+l}} \odot \mathcal{P}^{R^{1:M}_i}\right),
    \label{equ:calcLeft}
\end{equation}
and its batch computation for all $i=1,\ldots,l$ is referred to as the function $\texttt{calcLeft}()$ in related algorithms. While for subproblem 2, we compute:
\begin{equation}
    \mathcal{P}^{R^{1:M}_i} = \mathcal{P}^{P^{1:M}_{i+l}} \odot \left(\mathcal{P}^{L^{1:M}_i} \circledast' \mathcal{P}^{P^{1:M}_i}\right),
    \label{equ:calcRight}
\end{equation}
and the corresponding batch computation will be denoted as $\texttt{calcRight}()$. As for subproblem 3, the solution is:
\begin{equation}
    \begin{aligned}
        \mathcal{P}^{P^{1:M}_i} &= \mathcal{P}^{L^{1:M}_i} \circledast' \mathcal{P}^{R^{1:M}_i}, \\
        \mathcal{P}^{P^{1:M}_{i+l}} &= \mathcal{P}^{R^{1:M}_i},
    \end{aligned}
    \label{equ:calcParent}
\end{equation}
and the batch computation is denoted as $\texttt{calcParent}()$.

It is worth noting that the non-batched forms of equations \eqref{equ:calcLeft}, \eqref{equ:calcRight}, and \eqref{equ:calcParent} should be considered constant-time operations, as their complexity depends only on the base of joint distribution, which is a fixed parameter determined at the system design stage. A more detailed analysis reveals that the normalized combination operation in \eqref{equ:nrmcomb} has a traversal complexity, while the circular convolution operations in \eqref{equ:circonv} and \eqref{equ:circonv_inv} can be implemented using the fast Fourier transform.

\subsection{Inference over Polar Transforms}

With the above basic computational tools, we now describe the general solution for inference over polar transforms, which is based on the well-known probability propagation principles for graphical models. Note that every node in the decoding tree is either the root or a leaf of a three-node subtree, and these two cases correspond to different roles in the recursive procedures, as illustrated in \figref{fig:getAsFuncs}. Following the discussion in \cite{understand}, there are two types of message passing computations, corresponding to two recursive algorithms: $\texttt{getAsParent}()$ (\algref{alg:getAsParent}) and $\texttt{getAsChild}()$ (\algref{alg:getAsChild}).

Specifically, each node is indexed by an integer $\beta$, starting from the root $\beta = 1$ and proceeding top-down, left-to-right, as illustrated in \figref{fig:decTree}. The recursion terminates either at the root node (representing the source distribution or a conditional distribution with side information \cite{5513567}), or at a leaf node (corresponding to a known value or a uniform distribution). In the algorithms, the expression $\lfloor \beta / 2 \rfloor$ retrieves the parent node of the $\beta$-th node; $2\beta$ and $2\beta + 1$ access the left and right children of the $\beta$-th node, respectively; and the parity of $\beta$ is used to determine whether the node is a left or right child.

Returning to the SC decoding of monotone chain polar codes, the task corresponds to computing the following $MN$ conditional probabilities:
\begin{equation*}
    \Pr(U^{\gamma_t}_{i_t} \mid U^{\gamma_1}_{i_1}=u^{\gamma_1}_{i_1},\ldots,U^{\gamma_{t-1}}_{i_{t-1}}=u^{\gamma_{t-1}}_{i_{t-1}}),
\end{equation*}
which can be obtained by first computing the $MN$ joint distributions:
\begin{equation*}
    \Pr(U^{1:M}_{i_t} \mid U^{\gamma_1}_{i_1}=u^{\gamma_1}_{i_1},\ldots,U^{\gamma_{t-1}}_{i_{t-1}}=u^{\gamma_{t-1}}_{i_{t-1}}),
\end{equation*}
and then marginalizing over the $\gamma_t$-th component of the multivariate variable $U^{1:M}_{i_t}$.

\begin{figure}[t!]
    \centerline{\includegraphics[width=0.35\textwidth]{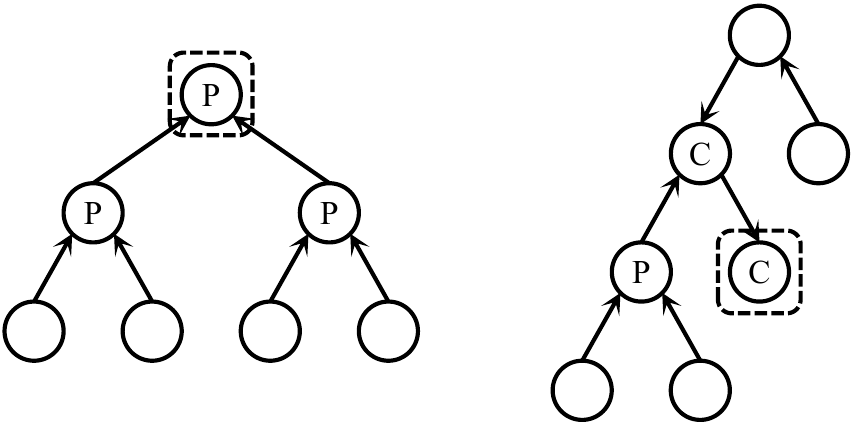}}
    \caption{Recursive structure of the two function calls (nodes labelled with P and C respectively).}
    \label{fig:getAsFuncs}
\end{figure}

\begin{algorithm}[t!]
    \caption{$\texttt{getAsParent}(\beta)$}
    \label{alg:getAsParent}
    \KwIn{node index: $\beta$}
    \KwOut{tensor array of the $\beta$-th node: \texttt{arr}}
    \If{$N \leq \beta \leq 2N - 1$}{
        \Return the tensor of this leaf node\;
    }
    Recursively get $\texttt{arr1} \gets \texttt{getAsParent}(2\beta)$\;
    Recursively get $\texttt{arr2} \gets \texttt{getAsParent}(2\beta + 1)$\;
    Compute $\texttt{arr} \gets \texttt{calcParent}(\texttt{arr1}, \texttt{arr2})$\;
\end{algorithm}

\begin{algorithm}[t!]
    \caption{$\texttt{getAsChild}(\beta)$}
    \label{alg:getAsChild}
    \KwIn{node index: $\beta$}
    \KwOut{tensor array of the $\beta$-th node: \texttt{arr}}
    \If{$\beta = 1$}{
        \Return the tensor array of root node\;
    }
    Recursively get $\texttt{arr1} \gets \texttt{getAsChild}(\lfloor \beta / 2 \rfloor)$\;
    \eIf{$\text{mod}(\beta, 2) = 0$}{
        Recursively get $\texttt{arr2} \gets \texttt{getAsParent}(\beta + 1)$\;
        Compute $\texttt{arr} \gets \texttt{calcLeft}(\texttt{arr1}, \texttt{arr2})$\;
    }{
        Recursively get $\texttt{arr2} \gets \texttt{getAsParent}(\beta - 1)$\;
        Compute $\texttt{arr} \gets \texttt{calcRight}(\texttt{arr1}, \texttt{arr2})$\;
    }
\end{algorithm}

Notably, at each decoding step, only a subset of the components within the multivariate variable is known, while the rest remain unknown. More precisely, before the $t$-th decoding step, for each multivariate variable $U^{1:M}_i$, only the following subset of component values is available:
\begin{equation}
    \mathcal{J}(t,i) \triangleq \{ \gamma_{t'} : t' < t,\ i_{t'} = i \}.
    \label{equ:partial_set}
\end{equation}

To accommodate this property, a natural choice is to perform \emph{partial decision}. Specifically, for each $i = 1,\ldots,N$, we set the probability tensor $\mathcal{P}^{U^{1:M}_i}$ at the corresponding leaf node to satisfy the following: for all $k_{\mathcal{J}(t,i)} \neq u_{\mathcal{J}(t,i)}$, we set $\mathcal{P}^{U^{1:M}_i}_{k_{1:M}} = 0$; and for all $k_{\mathcal{J}(t,i)} = u_{\mathcal{J}(t,i)}$ and $k'_{\mathcal{J}(t,i)} = u_{\mathcal{J}(t,i)}$, we ensure that $\mathcal{P}^{U^{1:M}_i}_{k_{1:M}} = \mathcal{P}^{U^{1:M}_i}_{k'_{1:M}}$.

Intuitively, such a distribution is \emph{partially uniform}, or equivalently, \emph{partially deterministic}: entries that contradict the known assignments are set to zero probability, while all remaining entries are assigned equal probability.

Under this setting, the joint distribution at step $t$ can be obtained by applying corresponding partial decision to each leaf node and then executing the function $\texttt{getAsChild}(\beta_t)$, where $\beta_t = i_t + N - 1$. The SC decoding task is formalized into a sequence of inference subtasks.

\subsection{Graph Representation of Computations}

\begin{figure}[t!]
    \centerline{\includegraphics[width=0.45\textwidth]{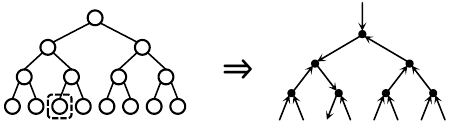}}
    \caption{Computational graph of the inference on $U^{1:M}_3$.}
    \label{fig:decGraph}
\end{figure}

From the above discussion, it is clear that depending on whether $\texttt{getAsParent}()$ or $\texttt{getAsChild}()$ is called, the probabilities associated with a tree node carries different meanings. To facilitate the subsequent discussion, we introduce the concept of the \emph{computational graph}, which provides an intuitive and structured representation of the underlying probability computations within each inference subtask.

An example is given in \figref{fig:decGraph} for a code of length $2^3$, where the decoding of a leaf node is mapped to a graph-based computation flow. Each computation is represented as a directed edge: calls to $\texttt{getAsParent}()$ and $\texttt{getAsChild}()$ correspond to upward and downward edges, respectively. These edges are connected through vertices according to the hierarchical structure of the original decoding binary tree. For a code of length $N = 2^n$, the computational graph consists of $2N-1$ edges and $N-1$ vertices.

It is important to emphasize a key distinction: although the computational graph resembles the decoding binary tree, the correspondence is not one-to-one. In the computational graph, each vertex corresponds to a three-node subtree in the original decoding tree, rather than a node representing a group of random variables. In fact, the entities associated with groups of random variables are the directed edges. At level $j$, each edge corresponds to a batch computation of complexity $O(N/2^j)$. A precise understanding of these edges and vertices is crucial for our proposed scheme.

\subsection{Time Complexity Analysis}

Returning to the inference algorithm, it is straightforward to verify that executing $\texttt{getAsChild}()$ on leaf nodes traverses all random variables exactly once, without redundancy or omission. The corresponding complexity is $O(N\log{N})$. If this procedure were executed independently for all $N$ leaves, the total complexity of SC decoding would be $O(N^2\log{N})$. In the absence of special structural properties, this complexity is essentially the best we can achieve. However, SC decoding does exhibit certain properties that enable further improvements. They can be categorized into two types:

\begin{itemize}
    \item \textbf{Intra-step sharing:} Among the $O(N\log{N})$ random variables, a considerable fraction may share identical probability distributions. This is most evident in undecoded nodes, all of which initially correspond to uniform distributions. Exploiting this property reduces the complexity of a single inference subtask to $O(N)$.
    \item \textbf{Inter-step sharing:} During the entire decoding process, different decoding steps share a substantial number of identical computations. Leveraging this property reduces the overall complexity of SC decoding.
\end{itemize}

In the classical work~\cite{5075875}, the final solution combines both strategies. In practice, however, the role of intra-step sharing is very limited. This is because the data subject to inter-step sharing always includes that of intra-step sharing. As a result, except for the first decoding step (where the complexity can be reduced from $O(N\log{N})$ to $O(N)$), subsequent steps remain unaffected. Since the overall complexity of SC decoding cannot be reduced below $O(N\log{N})$, we do not take intra-step sharing into consideration in subsequent analysis.

The total computational complexity of SC decoding is equivalent to the cumulative cost of all distinctive calls to $\texttt{getAsParent}()$ and $\texttt{getAsChild}()$ throughout the entire decoding process. We present the following propositions to provide a systematic description.

\begin{proposition}
    The time complexity of SC decoding for monotone chain polar codes is lower bounded by $O(N\log{N})$.
\end{proposition}

\begin{proof}
    This conclusion is obvious. By the end of a complete SC decoding process, all leaf nodes corresponding to the random variables $U^{1:M}_1,\ldots,U^{1:M}_N$ must transition from complete uncertainty to full determinism. Achieving this requires propagating the prior distribution from the root node to every leaf node, and the paths collectively cover all nodes in the binary tree without omission. It follows that the entire decoding tree, containing $O(N\log{N})$ random variables, must be traversed at least once. Therefore, the decoding process cannot be completed with a time complexity lower than $O(N\log{N})$.
\end{proof}

\begin{proposition}
    The time complexity of SC decoding for monotone chain polar codes is upper bounded by $O(N^2)$.
\end{proposition}

\begin{proof}
    We observe that the difference in computational cost between any two consecutive decoding steps is at most $O(N)$. This stems directly from the structural properties of cancellation decoding. Specifically, from step $t$ to step $t+1$, the only newly introduced information is the value of a single variable, $U^{\gamma_t}_{i_t} = u^{\gamma_t}_{i_t}$, while all other variables remain unchanged. As a result, the two corresponding inference tasks share a significant amount of identical computation, with the only difference being the propagation of probability from $U^{\gamma_t}_{i_t}$ to $U^{\gamma_{t+1}}_{i_{t+1}}$. In the computational graph, this difference intuitively corresponds to the path between the $i_t$-th and $i_{t+1}$-th leaf nodes.

    Specifically, we illustrate the computational graph for two consecutive inference subtasks as shown in \figref{fig:optimTime}, where the dashed edges represent identical computations. Since the edge at level $j$ involves $O(N/2^j)$ random variables, the total cost along such a path is at most $O(N)$, as it may traverse all levels $j = 1, \ldots, \log_2{N}$ in the worst case. Repeating this process over $MN$ decoding steps results in an overall decoding complexity of at most $O(N^2)$.
\end{proof}

\begin{figure}[t!]
    \centerline{\includegraphics[width=0.4\textwidth]{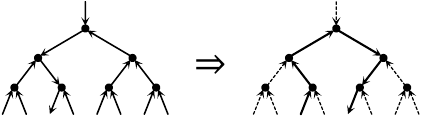}}
    \caption{Inter-step sharing between decoding steps $t$ and $t+1$, where $i_t=3$ and $i_{t+1}=5$.}
    \label{fig:optimTime}
\end{figure}

We further present two propositions demonstrating the tightness of the lower and upper bounds:

\begin{proposition}
    There exists a class of monotone chains whose SC decoding complexity is $O(N\log{N})$.
\end{proposition}

\begin{proof}
    It suffices to provide one such class. Consider the monotone chain corresponding to a corner point of the SW rate region, given by:
    \begin{equation*}
        \gamma_{1:MN} = \underbrace{1,\ldots,1}_{N},\ldots,\underbrace{M,\ldots,M}_{N},
    \end{equation*}
    which decomposes the decoding task into $M$ classical SC decoding tasks. The corresponding decoding index $i_t$ increases sequentially from $1$ to $N$ and repeat $M$ times. As a result, each edge in the computation graph is updated exactly $2M$ times, leading to an overall complexity of $O(N\log{N})$.
\end{proof}

\begin{proposition}
    There exists a class of monotone chains whose SC decoding complexity reaches $O(N^2)$.
\end{proposition}

\begin{proof}
    Consider the following monotone chain for $M = 2$:
    \begin{equation*}
        \gamma_{1:2N} = \underbrace{1,\ldots,1}_{N/2},\underbrace{1,2,\ldots,1,2}_{N},\underbrace{2,\ldots,2}_{N/2},
    \end{equation*}
    which first decodes $U^1_{1:N/2}$ sequentially, then alternates between $U^1_{N/2+i}$ and $U^2_i$ for $i = 1,\ldots,N/2$, and finally completes the decoding of $U^2_{(N/2+1):N}$.

    The key aspect of this chain lies in its alternating region. Each switch occurs between the left and right subtrees of the root node, meaning that the probability propagation between $U^1_{N/2+i}$ and $U^2_i$ must always traverse the entire depth of the tree. Therefore, each switch incurs the maximum cost of $O(N)$, resulting in an overall complexity of at least $O(N^2)$. Recalling the upper bound established previously, this confirms that the total complexity is indeed $O(N^2)$.
\end{proof}

In summary, the low decoding complexity of classical polar codes is primarily due to the successive order from $i = 1$ to $N$. In contrast, general monotone chains often involve frequent jumps between indices $i_t \in \{1, \ldots, N\}$. The complexity overhead of each jump ranges from constant to $O(N)$, leading to an overall SC decoding complexity vary from $O(N\log{N})$ to $O(N^2)$.

\subsection{Space Complexity Analysis}

Thanks to the recursive structure of the polar transform, even without any specific memory optimization, the space complexity remains relatively low at $O(N\log{N})$. Moreover, in principle, it is also possible to reduce the space complexity to $O(N)$ in all cases. This can be achieved by replacing the recursive computation with an iterative procedure that processes the decoding tree layer by layer. Since each layer only requires $O(N)$ memory to store intermediate results, the total space usage can be significantly reduced. However, this approach has a major drawback: it only supports intra-step sharing, leading to an overall time complexity of $O(N^2)$.

A more practical optimization strategy was proposed in \cite{7055304}, which leverages the successive decoding order of classical polar codes. As illustrated in \figref{fig:optimSpace}, at decoding step $i_t = t$, the downward-directed edges form a single path from the root edge to the corresponding leaf edge. This effectively partitions the entire computational graph into two parts: deterministic part and uniform part. This is because the values of the previously decoded variables $U_{1:t-1}$ are already known, and the values of the future variables $U_{t+1:N}$ remain completely unknown, following from \equref{equ:calcParent} that deterministic inputs will produce deterministic outputs, and uniform inputs will yield uniform outputs.

Furthermore, due to the increasing decoding order $i_t = t$ of classical SC decoding, we can draw two key conclusions:
\begin{itemize}
    \item Certain edges in the deterministic subgraph will never be used again in future decoding steps and can therefore be safely released;
    \item Certain edges in the uniform subgraph will not be accessed until a specific future step, and thus do not need to be stored beforehand.
\end{itemize}

From the structure of the graph, it is clear that only the edges directly connected to downward edges, as well as the downward edges themselves, need to be retained. This corresponds to $O(N)$ probabilities and $O(\log{N})$ pointers, which is evident since at most two edges are kept at each layer. Among these retained edges, those located in the deterministic subgraph and the others respectively correspond to the bit arrays and probability arrays used in the memory-efficient SC decoder proposed in \cite{7055304}.

Based on our previous discussion, we know that the index sequence $i_{1:MN}$ of a general monotone chain may contain frequent jumps. As a result, most monotone chain polar codes cannot benefit from the aforementioned classical space optimization strategy. Nonetheless, this does not prevent list decoding algorithms from being efficiently implemented for such codes. In the next section, we will introduce a constant-complexity forking strategy that enables efficient list decoding without relying on lazy-copy method.

\begin{figure}[t!]
    \centerline{\includegraphics[width=0.4\textwidth]{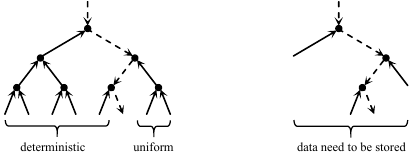}}
    \caption{Space optimization strategy adopted in classical SC decoding schemes.}
    \label{fig:optimSpace}
\end{figure}

\section{Proposed Time-Efficient Algorithms}

In this section, we present the implementation details of a series of algorithms that ultimately lead to a time-efficient SC list decoding method for general monotone chain polar codes.

\subsection{Time-Efficient SC Decoder}

\begin{figure*}[t!]
    \centerline{\includegraphics[width=\textwidth]{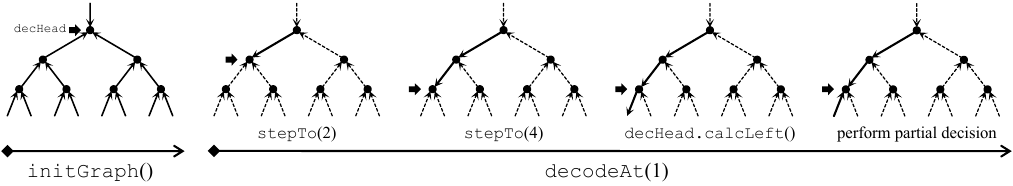}}
    \caption{Illustration of the guided traversal process.}
    \label{fig:walkProcess}
\end{figure*}

We begin by introducing a time-efficient SC decoder for general monotone chain polar codes, which builds upon the computational graph framework.

We first describe the implementation of the computational graph, which consists of two core data structures: \texttt{Edge} and \texttt{Vertex}. Each \texttt{Edge} includes a pointer \texttt{from} referencing its starting vertex, and a \texttt{data} array that stores the associated probability tensors. Each \texttt{Vertex} includes three pointers \texttt{parent}, \texttt{left}, and \texttt{right}, pointing to the corresponding edges. It also provides three batch-processing functions: $\texttt{calcLeft}()$, $\texttt{calcRight}()$, and $\texttt{calcParent}()$, as detailed in \equref{equ:calcLeft}, \equref{equ:calcRight}, and \equref{equ:calcParent}. These functions implicitly read from and write to the connected edges of this vertex.

To distinguish between different edges and vertices, we adopt a binary-tree-like indexing scheme for the computational graph. Indices start from $1$ at the top level and proceed level-by-level from left to right. This scheme provides convenient relationships, such as: for vertex $\beta$, its associated edges are indexed by $\beta$, $2\beta$, and $2\beta+1$; for edge $\beta$, the connected vertices are indexed by $\lfloor\beta/2\rfloor$ and $\beta$. Moreover, for decoding step $t$ at position $i_t$, the corresponding edge is $(i_t + N - 1)$, and its associated vertex is $\lfloor(i_t + N - 1)/2\rfloor$.

Returning to the decoding algorithm, the core idea is to model the SC decoding process as a guided traversal of a pointer \texttt{decHead} across the computational graph. An intuitive illustration is provided in \figref{fig:walkProcess}, which demonstrates the behavior of \algref{alg:initGraph} and \algref{alg:decodeAt}. Overall, the decoding process can be summarized as follows:

\begin{enumerate}
    \item Initialize the graph and the head pointer \texttt{decHead};
    \item Identify the target vertex for the next decoding step, and get the path starting from \texttt{decHead}. Then traverse the path with \texttt{decHead}, and update the corresponding edges along the path;
    \item Decode the target variable by hard decision or frozen value, then perform a partial decision;
    \item If decoding is not yet complete, return to step 2.
\end{enumerate}

In more detail, \algref{alg:decodeAt} consists of two main stages. In the updating stage, we move \texttt{decHead} along the path generated by \algref{alg:getPath} to the target vertex $\lfloor(i_t + N - 1)/2\rfloor$, which connects to the $i_t$-th leaf edge. During traversal, the relevant edges are updated using \algref{alg:stepTo}, and the target vertex finally invokes either \texttt{calcLeft}() or \texttt{calcRight}() to obtain the target conditional distribution:
\begin{equation*}
    \Pr(U^{1:M}_{i_t} \mid U^{\gamma_1}_{i_1} = u^{\gamma_1}_{i_1}, \ldots, U^{\gamma_{t-1}}_{i_{t-1}} = u^{\gamma_{t-1}}_{i_{t-1}}).
\end{equation*}

While for the decision stage, we update the $i_t$-th leaf edge with a partially deterministic tensor based on known frozen values or previously decoded values. The detailed description of partial decision is presented previously in \equref{equ:partial_set}.

These algorithms collectively constitute the proposed SC decoder, which achieves time efficiency by fully leveraging the redundancy in computations across different decoding steps. It is important to note that the primary computational load arises from updates to the probability tensor arrays, performed within the $\texttt{stepTo}()$ function. At tree level $j$, this operation incurs a complexity of $O(N/2^j)$. As for the auxiliary function $\texttt{getPath}()$ which determines the traversal path of the head pointer, it has a complexity of at most $O(\log{N})$ and therefore does not impact the overall time complexity.

\begin{algorithm}[t!]
    \caption{$\texttt{initGraph}()$}
    \label{alg:initGraph}
    Allocate memory for the computational graph\;
    Set the \texttt{data} attribute of the first edge with the prior probability distribution of $X^{1:M}_1,\ldots,X^{1:M}_N$\;
    \For{$i=2$ \textbf{to} $2N-1$}{
        Set \texttt{data} of the $i$-th edge to be uniform\;
    }
    Set relationships between edges and vertices properly\;
\end{algorithm}

\begin{algorithm}[t!]
    \caption{$\texttt{decodeAt}(t)$}
    \label{alg:decodeAt}
    \KwIn{the current decoding step: $t$}
    Get the target index $\beta_t = \lfloor (i_t + N - 1) / 2 \rfloor$\;
    Generate $\texttt{path} \gets \texttt{getPath}(\beta_t)$\;
    \lForEach{$\beta$ \textbf{in} \texttt{path}}{$\texttt{stepTo}(\beta)$}
    \eIf{$\text{mod}(i_t, 2) = 1$} {
        Denote \texttt{decHead}.\texttt{left} as \texttt{leaf}\;
        Temporarily store \texttt{leaf}.\texttt{data} in \texttt{temp}\;
        Execute $\texttt{decHead}.\texttt{calcLeft}()$\;
    } {
        Denote \texttt{decHead}.\texttt{right} as \texttt{leaf}\;
        Temporarily store \texttt{leaf}.\texttt{data} in \texttt{temp}\;
        Execute $\texttt{decHead}.\texttt{calcRight}()$\;
    }
    \eIf{$U^{\gamma_t}_{i_t}$ is frozen} {
        Get the frozen value $u^{\gamma_t}_{i_t}$\;
    } {
        Combine \texttt{leaf}.\texttt{data} with \texttt{temp} using \equref{equ:nrmcomb}\;
        Estimate $u^{\gamma_t}_{i_t}$ based on \texttt{leaf}.\texttt{data}\;
    }
    Set \texttt{leaf}.\texttt{data} to be partially deterministic at $U^{\mathcal{J}(t+1, i_t)}_{i_t} = u^{\mathcal{J}(t+1, i_t)}_{i_t}$, where $\mathcal{J}$ is defined in \equref{equ:partial_set}\;
\end{algorithm}

\begin{algorithm}[t!]
    \caption{$\texttt{getPath}(\beta_t)$}
    \label{alg:getPath}
    \KwIn{Index of the target vertex: $\beta_t$}
    \KwOut{The resulting index array: \texttt{path}}
    Denote the current index of \texttt{decHead} as $\beta'$\;
    \While{$\beta' \neq \beta_t$} {
        \eIf{$\beta' < \beta_t$} {
            Append $\beta'$ to \texttt{path} from the head\;
            Let $\beta' \gets \lfloor \beta' / 2 \rfloor$\;
        } {
            Append $\beta_t$ to \texttt{path} from the tail\;
            Let $\beta_t \gets \lfloor \beta_t / 2 \rfloor$\;
        }
    }
    Append the common root of $\beta'$ and $\beta_t$ to \texttt{path}\;
\end{algorithm}

\begin{algorithm}[t!]
    \caption{$\texttt{stepTo}(\beta)$}
    \label{alg:stepTo}
    \KwIn{Index of the target vertex: $\beta$}
    Let the current index of \texttt{decHead} be $\beta'$\;
    \lIf{$\beta'$ is not a neighbor of $\beta$}{\Return}
    \eIf{$\beta' = \lfloor \beta / 2 \rfloor$} {
        \eIf{$\text{mod}(\beta, 2) = 0$} {
            Execute $\texttt{decHead}.\texttt{calcLeft}()$\;
            Swap $\texttt{decHead}$ and \texttt{decHead}.\texttt{left}.\texttt{from}\;
        } {
            Execute $\texttt{decHead}.\texttt{calcRight}()$\;
            Swap $\texttt{decHead}$ and \texttt{decHead}.\texttt{right}.\texttt{from}\;
        }
    } {
        Execute $\texttt{decHead}.\texttt{calcParent}()$\;
        Swap $\texttt{decHead}$ and \texttt{decHead}.\texttt{parent}.\texttt{from}\;
    }
\end{algorithm}

\subsection{Time-Efficient Decoder Forking}

From the preceding discussion, it is clear that the state of a directed edge, whether upward or downward, plays a critical role. Conceptually, each edge should have both a head and a tail vertex. In our implementation, however, only the \texttt{from} pointer is retained. This is because the computational graph is always maintained in a special state during decoding, where all directed edges point toward \texttt{decHead}. Under this condition, neither the explicit state nor the tail vertex of an edge needs to be stored.

This structural property brings a clear advantage. Starting from the vertex pointed to by \texttt{decHead}, one can traverse all its adjacent edges and, through each edge's \texttt{from} pointer, reach other vertices. By prohibiting any backtracking to previously visited vertices, this traversal recursively visits every vertex and edge in the computational graph. In this sense, \texttt{decHead} captures the complete internal state of the SC decoder.

Building on this fact, we can create two new head pointers based on an existing one, denoted as \texttt{decHead1} and \texttt{decHead2}, as shown in \figref{fig:exploreAt}. These correspond to two SC decoders that share the entire graph except for a single adjacent edge. Applying the above \texttt{from}-pointer traversal starting from \texttt{decHead1} and \texttt{decHead2} produces two decoder instances whose internal structures differ only in that one edge. Therefore, creating a new head pointer is equivalent to instantiating a new SC decoder, and such duplication can be achieved at constant computational cost.

This efficiency comes from avoiding global indexing for vertex and edge access. All access is done through local pointers rooted at \texttt{decHead}. This naturally leads to the concept of a \emph{logical computational graph}, uniquely determined by its head pointer. The graph data is not stored contiguously in memory but is organized through local pointer references, with large parts shared among different decoders.

\begin{figure}[t!]
    \centerline{\includegraphics[width=0.45\textwidth]{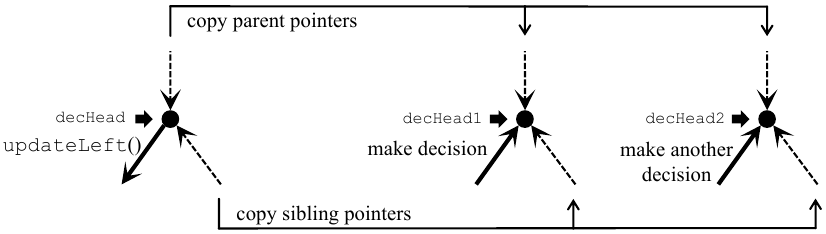}}
    \caption{Constant-time decoder forking based on head pointer.}
    \label{fig:exploreAt}
\end{figure}

\subsection{Time-Efficient SC List Decoding}

\subsubsection{An Initial Attempt}
In the proposed SC decoder, a complete computational graph is initialized, and all subsequent computation results are stored directly in the initially allocated memory. When extending this approach to list decoding, a minor adjustment is required: multiple SC decoders generated via the proposed decoder-forking strategy share the same memory space, which may lead to access conflicts.

A straightforward solution is to have $\texttt{stepTo}()$ always allocate and operate on fresh memory. Specifically, at each invocation of $\texttt{stepTo}()$, instead of overwriting the probability tensor of an existing edge, we allocate a new edge to store the updated value. Furthermore, rather than updating \texttt{decHead} to point directly to the next vertex, we employ the forking mechanism to generate a new vertex instance for the next vertex, and point \texttt{decHead} to it. The newly created vertex–edge pair is then embedded into the original graph, replacing the existing one. Given that the pointers to the old memory blocks are discarded, this approach ensures that each \texttt{decHead} instance exclusively references the most recent and valid data, thereby preserving decoding consistency. Moreover, since different head pointers are assigned to distinct new memory blocks, access conflicts between SC decoders are inherently avoided.

One remaining issue is that if discarded memories are not recycled promptly, the space complexity will match the time complexity. For monotone chains with time complexity up to $O(N^2)$, this leads to substantial and unnecessary memory overhead. A less obvious issue is that repeatedly releasing these blocks can cause severe memory fragmentation. Efficient memory management is therefore essential to overcome these limitations.

\subsubsection{Memory Management}
In the classical implementation proposed by \cite{7055304}, memory management relies on an explicit path-killing operation. In our design, we adopt a more compact and efficient scheme that eliminates the need for explicit memory release.

The principle is overwriting on existing space. During initialization, memory is preallocated for $L$ computational graphs, requiring $O(LN\log N)$ space in total. During decoding, let $L' \in \{1,\ldots,L\}$ denote the number of currently active SC decoders. At each invocation of $\texttt{stepTo}()$, we simply mark $L'$ of the preallocated graphs as writable and assign them to the $L'$ head pointers. Naturally, in implementing this strategy, some of the writable blocks may contain data that are still required at the current step; such data can be temporarily stored, and the additional space is only of $O(N)$. In essence, no new memory is allocated, as existing blocks are reused by directly overwriting their contents.

\subsubsection{Likelihoods Calculation}
After list decoding, the final decision is made by selecting the most likely candidate from the list of $L$ codewords. This requires computing the joint probability of the reconstructed variables, which can be derived using the chain rule:
\begin{equation}
    \begin{aligned}
         &\Pr(X^{1:M}_1 = x^{1:M}_1, \ldots, X^{1:M}_N = x^{1:M}_N) \\
         &= \prod_{t = 1}^{MN} \Pr(U^{\gamma_t}_{i_t} = u^{\gamma_t}_{i_t} \mid U^{\gamma_1}_{i_1} = u^{\gamma_1}_{i_1}, \ldots, U^{\gamma_{t-1}}_{i_{t-1}} = u^{\gamma_{t-1}}_{i_{t-1}}).
    \end{aligned}
\end{equation}

To ensure numerical stability, the likelihood is typically computed in the logarithmic domain. During decoding, we maintain a log-likelihood value for each \texttt{decHead} and update it incrementally. For convenience, we may use natural logarithms $\log(\cdot)$ for all $U^{\gamma_t}_{i_t}$.

\subsubsection{Get the Decoding Result}
Among the $L$ SC decoders, the one with the highest likelihood corresponds to the final decoding result. Although the decoded codeword can be reconstructed via the previously described pointer-based graph traversal, the complexity is of $O(N\log{N})$. Here we introduce a more efficient method with complexity $O(N)$.

At the end of decoding, the corresponding head pointer always points to the $(N-1)$-th vertex. We move it sequentially to the root vertex by $\texttt{stepTo}()$, and finally invoke $\texttt{calcParent}()$. The resulting tensors in the root edge yield the joint distributions of $X^{1:M}_1, \ldots, X^{1:M}_N$, which are deterministic. The values with probability one constitute the final decoding result $x^{1:M}_1,\ldots,x^{1:M}_N$.

\section{Numerical Results}

We first provide a numerical illustration of the complexity improvement achieved by the proposed constant-time decoder forking strategy. Consider a distributed source coding scenario with two correlated binary sources, whose joint distribution is [0.1286, 0.0175, 0.0175, 0.8364] for $X^{1,2}=00,01,10,11$. Experiments are performed along the corner chain $1,\ldots,1,2,\ldots,2$ (with SC decoding complexity $O(N \log N)$) using a list size of $L=2$ and an empty frozen set. Table~\ref{tab:tab01} reports the average runtime through 100 simulation rounds, in comparison with the classical lazy-copy strategy proposed in \cite{7055304}. During the experiments, no fast-pruning techniques~\cite{7339671,8669947} or binary-specific optimizations, such as the log-likelihood-based min-sum algorithm, were employed, as they are beyond the scope of this work. Our focus is on the relative complexity between the two strategies, not the exact values.

The results reveal a notable trend. For short blocklengths, the performance gain of the proposed scheme is not very significant, as both methods exhibit nearly linear runtime growth. For longer blocklengths, however, the proposed scheme continues to scale linearly, while the lazy-copy strategy incurs a much sharper increase. This is because the complexity introduced by decoder forking has a relatively small constant factor compared with the probability-computation part of decoding. Thus, the advantage of the proposed scheme becomes evident at moderate to large blocklengths.

\begin{table}[t!]
    \centering
    \caption{Runtime comparison (in seconds).}
    \label{tab:tab01}
    \resizebox{0.45\textwidth}{!}{
    \begin{tabular}{r|l l l}
        \toprule
        \addlinespace[0.5ex]
        Blocklength & Lazy-copy & \multicolumn{2}{l}{Proposed} \\
        \midrule
        \addlinespace[0.5ex]
        64      & $5.87\times 10^{-3}$ & $5.79\times 10^{-3}$ & $(-1.4\%)$ \\
        \addlinespace[0.3ex]
        256     & $2.46\times 10^{-2}$ & $2.41\times 10^{-2}$ & $(-2.0\%)$  \\
        \addlinespace[0.3ex]
        1024    & $1.03\times 10^{-1}$ & $9.49\times 10^{-2}$ & $(-5.4\%)$  \\
        \addlinespace[0.3ex]
        4096    & $4.29\times 10^{-1}$ & $3.75\times 10^{-1}$ & $(-12.6\%)$ \\
        \addlinespace[0.3ex]
        16384   & $2.51\times 10^{0}$  & $1.44\times 10^{0}$  & $(-42.6\%)$ \\
        \addlinespace[0.3ex]
        65536   & $2.02\times 10^{1}$  & $5.43\times 10^{0}$  & $(-73.1\%)$ \\
        \addlinespace[0.3ex]
        262144  & $2.55\times 10^{2}$  & $2.18\times 10^{1}$  & $(-91.5\%)$ \\
        \addlinespace[0.5ex]
        \bottomrule
    \end{tabular}}
\end{table}

We then evaluate the rate–distortion performance of SC and SC list decoding for monotone chain polar codes. Note that in distributed coding with $M$ terminals, the code rate is an $M$-dimensional vector $R_{1:M}$ rather than a single scalar. For clarity, we start from $H(X^{1:M})$ and increase the sum-rate while maintaining a fixed rate ratio across terminals. This ratio is determined by the chain rate defined in~\cite{6284254}:
\begin{equation*}
    R_{\gamma} = \frac{1}{N}\sum_{t=1,\gamma_t=\gamma}^{MN} H(U^{\gamma_t}_{i_t} | U^{\gamma_1}_{i_1},\ldots,U^{\gamma_{t-1}}_{i_{t-1}}),
\end{equation*}
for each terminal $\gamma = 1,\ldots,M$.

We consider a distributed source coding scenario with two correlated non-binary sources: a ternary source $X^1$ and a quinary source $X^2$, whose joint distribution is [0.0814, 0.6078, 0.0519, 0.0014, 0.0014, 0.0095, 0.0308, 0.0013, 0.0027, 0.0044, 0.0018, 0.0156, 0.0500, 0.0012, 0.1388] corresponding to $X^{1,2} = 00, 01, \ldots, 24$. We evaluate performance along two different monotone chains for short blocklength $N = 2^6$ and medium blocklength $N = 2^{10}$, and plot the block error rate (BLER) curves for list sizes $L = 1$ (SC decoding) and $L = 32$ (SC list decoding). Codes are constructed under 100 Monte-Carlo simulations. The results shown in \figref{fig:figure02} are obtained through 2000 experiments, where the two chains considered are the corner chain $1,\ldots,1,2,\ldots,2$ and a randomly generated monotone chain under $N=2^6$. While for $N=2^{10}$, we use the corresponding $4$-extension chain.

It can be observed that increasing the list size provides a substantial improvement in decoding performance, confirming that list decoding plays a critical role for monotone chain polar codes. Moreover, by comparing the two figures, it also reveals that better BLER performance might be achieved by other monotone chains than corner chains. Investigating the performance of monotone chain polar codes in general is therefore of considerable interest, further underscoring the importance of our work.

\begin{figure}[t!]
    \centering
    \subfloat[Corner chain.]{%
        \includegraphics[width=0.23\textwidth]{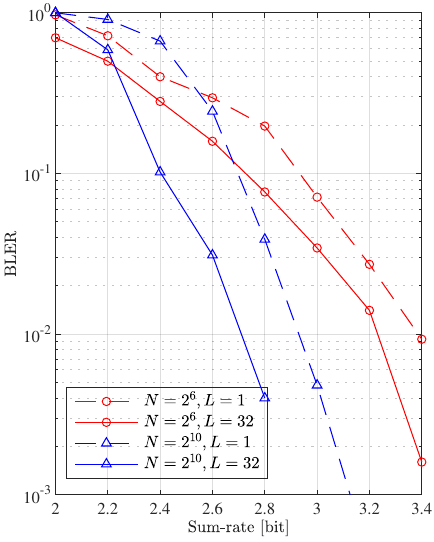}
    }
    \subfloat[Random chain.]{%
        \includegraphics[width=0.23\textwidth]{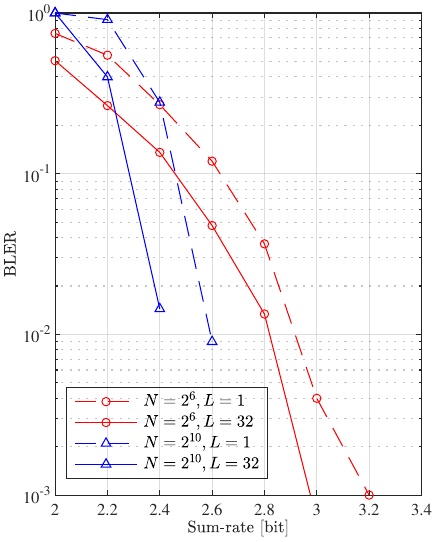}
    }
    \caption{BLER performance of SC and SC list decoding along difference monotone chains.}
    \label{fig:figure02}
\end{figure}

\section{Conclusion}

In this work, we investigated monotone chain polar codes and their SC decoding problems. We first established a unified SC decoding framework applicable to arbitrary numbers of terminals, arbitrary alphabets, and arbitrary monotone chains, and then provided a detailed complexity analysis. By introducing the concept of the computational graph, we showed that the time complexity of SC decoding along monotone chains can range from $O(N \log N)$ to $O(N^2)$, depending on the specific chain structure, and clarified the inherent limitations of classical lazy-copy strategy based on space-efficient techniques. Building upon this framework, we proposed efficient data structures and a constant-time decoder forking strategy, enabling time-efficient SC list decoding for general monotone chains. We believe that the tools developed in this work provide a foundation for further theoretical studies and practical applications of monotone chain polar codes.

\bibliographystyle{IEEEtran}
\bibliography{refs}

\begin{thebibliography}{10}
\providecommand{\url}[1]{#1}
\csname url@samestyle\endcsname
\providecommand{\newblock}{\relax}
\providecommand{\bibinfo}[2]{#2}
\providecommand{\BIBentrySTDinterwordspacing}{\spaceskip=0pt\relax}
\providecommand{\BIBentryALTinterwordstretchfactor}{4}
\providecommand{\BIBentryALTinterwordspacing}{\spaceskip=\fontdimen2\font plus
\BIBentryALTinterwordstretchfactor\fontdimen3\font minus \fontdimen4\font\relax}
\providecommand{\BIBforeignlanguage}[2]{{%
\expandafter\ifx\csname l@#1\endcsname\relax
\typeout{** WARNING: IEEEtran.bst: No hyphenation pattern has been}%
\typeout{** loaded for the language `#1'. Using the pattern for}%
\typeout{** the default language instead.}%
\else
\language=\csname l@#1\endcsname
\fi
#2}}
\providecommand{\BIBdecl}{\relax}
\BIBdecl

\bibitem{5075875}
E.~Arıkan, ``Channel polarization: A method for constructing capacity-achieving codes for symmetric binary-input memoryless channels,'' \emph{IEEE Trans. Inf. Theory}, vol.~55, no.~7, pp. 3051--3073, Jul. 2009.

\bibitem{5513567}
------, ``Source polarization,'' in \emph{IEEE Int. Symp. Inf. Theory}, Jun. 2010, pp. 899--903.

\bibitem{5437372}
S.~B. Korada and R.~L. Urbanke, ``Polar codes are optimal for lossy source coding,'' \emph{IEEE Trans. Inf. Theory}, vol.~56, no.~4, pp. 1751--1768, Apr. 2010.

\bibitem{6874846}
A.~G. Sahebi and S.~S. Pradhan, ``Polar codes for some multi-terminal communications problems,'' in \emph{IEEE Int. Symp. Inf. Theory}, Jun. 2014, pp. 316--320.

\bibitem{1055037}
D.~Slepian and J.~Wolf, ``Noiseless coding of correlated information sources,'' \emph{IEEE Trans. Inf. Theory}, vol.~19, no.~4, pp. 471--480, Jul. 1973.

\bibitem{5454148}
E.~Abbe and E.~Telatar, ``{MAC} polar codes and matroids,'' in \emph{2010 Inform. Theory and App. Workshop}, Jan. 2010, pp. 1--8.

\bibitem{5503184}
E.~Şaşoğlu, E.~Telatar, and E.~Yeh, ``Polar codes for the two-user binary-input multiple-access channel,'' in \emph{IEEE Inf. Theory Workshop Inf. Theory}, Jan. 2010, pp. 1--5.

\bibitem{6208869}
E.~Abbe and E.~Telatar, ``Polar codes for the $m$-user multiple access channel,'' \emph{IEEE Trans. Inf. Theory}, vol.~58, no.~8, pp. 5437--5448, Aug. 2012.

\bibitem{6284254}
E.~A. Bilkent, ``Polar coding for the {Slepian-Wolf} problem based on monotone chain rules,'' in \emph{IEEE Int. Symp. Inf. Theory}, Jul. 2012, pp. 566--570.

\bibitem{6620401}
S.~Önay, ``Successive cancellation decoding of polar codes for the two-user binary-input {MAC},'' in \emph{IEEE Int. Symp. Inf. Theory}, Jul. 2013, pp. 1122--1126.

\bibitem{7282710}
S.~Salamatian, M.~Médard, and E.~Telatar, ``A successive description property of monotone-chain polar codes for {Slepian-Wolf} coding,'' in \emph{IEEE Int. Symp. Inf. Theory}, Jun. 2015, pp. 1522--1526.

\bibitem{613189}
B.~Rimoldi and R.~Urbanke, ``Asynchronous {Slepian-Wolf} coding via source-splitting,'' in \emph{IEEE Int. Symp. Inf. Theory}, Jun. 1997, pp. 271--.

\bibitem{7055304}
I.~Tal and A.~Vardy, ``List decoding of polar codes,'' \emph{IEEE Trans. Inf. Theory}, vol.~61, no.~5, pp. 2213--2226, May 2015.

\bibitem{quantum}
\BIBentryALTinterwordspacing
C.~Hirche, ``Polar codes in quantum information theory,'' \emph{arXiv preprint arXiv:1501.03737}, Jan. 2015. [Online]. Available: \url{https://arxiv.org/abs/1501.03737}
\BIBentrySTDinterwordspacing

\bibitem{10578043}
Y.~Fang, ``Bridging hamming distance spectrum with coset cardinality spectrum for overlapped arithmetic codes,'' \emph{IEEE Trans. Inf. Theory}, vol.~70, no.~9, pp. 6714--6745, Sept. 2024.

\bibitem{5351487}
E.~Şaşoğlu, E.~Telatar, and E.~Arıkan, ``Polarization for arbitrary discrete memoryless channels,'' in \emph{IEEE Inf. Theory Workshop}, Oct. 2009, pp. 144--148.

\bibitem{6283740}
E.~Şaşoğlu, ``Polar codes for discrete alphabets,'' in \emph{IEEE Int. Symp. Inf. Theory}, Jul. 2012, pp. 2137--2141.

\bibitem{7456294}
R.~Nasser and E.~Telatar, ``Polar codes for arbitrary {DMCs} and arbitrary {MACs},'' \emph{IEEE Trans. Inf. Theory}, vol.~62, no.~6, pp. 2917--2936, Jun. 2016.

\bibitem{stack}
K.~Niu and K.~Chen, ``Stack decoding of polar codes,'' \emph{Electron. Lett.}, vol.~48, no.~12, pp. 695--697, Jun. 2012.

\bibitem{7094848}
O.~Afisiadis, A.~Balatsoukas-Stimming, and A.~Burg, ``A low-complexity improved successive cancellation decoder for polar codes,'' in \emph{Asilomar Conf. Signals, Syst. Comput.}, Nov. 2014, pp. 2116--2120.

\bibitem{6560025}
K.~Chen, K.~Niu, and J.~Lin, ``Improved successive cancellation decoding of polar codes,'' \emph{IEEE Trans. Commun.}, vol.~61, no.~8, pp. 3100--3107, Aug. 2013.

\bibitem{6355936}
B.~Li, H.~Shen, and D.~Tse, ``An adaptive successive cancellation list decoder for polar codes with cyclic redundancy check,'' \emph{IEEE Commun. Lett.}, vol.~16, no.~12, pp. 2044--2047, Dec. 2012.

\bibitem{6297420}
K.~Niu and K.~Chen, ``{CRC}-aided decoding of polar codes,'' \emph{IEEE Commun. Lett.}, vol.~16, no.~10, pp. 1668--1671, Oct. 2012.

\bibitem{7862172}
Q.~Zhang, A.~Liu, X.~Pan, and K.~Pan, ``{CRC} code design for list decoding of polar codes,'' \emph{IEEE Commun. Lett.}, vol.~21, no.~6, pp. 1229--1232, Jun. 2017.

\bibitem{7339671}
G.~Sarkis, P.~Giard, A.~Vardy, C.~Thibeault, and W.~J. Gross, ``Fast list decoders for polar codes,'' \emph{IEEE J. Sel. Areas Commun.}, vol.~34, no.~2, pp. 318--328, Feb. 2016.

\bibitem{8669947}
M.~H. Ardakani, M.~Hanif, M.~Ardakani, and C.~Tellambura, ``Fast successive-cancellation-based decoders of polar codes,'' \emph{IEEE Trans. Commun.}, vol.~67, no.~7, pp. 4562--4574, Jul. 2019.

\bibitem{9770084}
M.~C. Coşkun and H.~D. Pfıster, ``An information-theoretic perspective on successive cancellation list decoding and polar code design,'' \emph{IEEE Trans. Inf. Theory}, vol.~68, no.~9, pp. 5779--5791, Sept. 2022.

\bibitem{8396299}
A.~Elkelesh, M.~Ebada, S.~Cammerer, and S.~ten Brink, ``Belief propagation list decoding of polar codes,'' \emph{IEEE Commun. Lett.}, vol.~22, no.~8, pp. 1536--1539, Aug. 2018.

\bibitem{9174118}
H.~Yao, A.~Fazeli, and A.~Vardy, ``List decoding of {Arıkan's} {PAC} codes,'' in \emph{IEEE Int. Symp. Inf. Theory}, Jun. 2020, pp. 443--448.

\bibitem{9328621}
M.~Rowshan, A.~Burg, and E.~Viterbo, ``Polarization-adjusted convolutional ({PAC}) codes: Sequential decoding vs list decoding,'' \emph{IEEE Trans. Veh. Technol.}, vol.~70, no.~2, pp. 1434--1447, Feb. 2021.

\bibitem{6852102}
K.~Niu, K.~Chen, J.~Lin, and Q.~T. Zhang, ``Polar codes: Primary concepts and practical decoding algorithms,'' \emph{IEEE Commun. Mag.}, vol.~52, no.~7, pp. 192--203, Jul. 2014.

\bibitem{butterfly}
\BIBentryALTinterwordspacing
E.~Arıkan, ``Entropy polarization in butterfly transforms,'' \emph{Digit. Signal Process.}, vol. 119, p. 103207, Dec. 2021. [Online]. Available: \url{https://www.sciencedirect.com/science/article/pii/S1051200421002463}
\BIBentrySTDinterwordspacing

\bibitem{understand}
J.~S. Yedidia, W.~T. Freeman, and Y.~Weiss, ``Understanding belief propagation and its generalizations,'' in \emph{Exploring Artificial Intelligence in the New Millennium}, G.~Lakemeyer and B.~Nebel, Eds.\hskip 1em plus 0.5em minus 0.4em\relax San Francisco, CA, USA: Morgan Kaufmann, 2003, pp. 239--269.

\end{thebibliography}

\end{document}